\documentclass[11pt]{article}
\usepackage{amsmath,amsthm}
\usepackage{amssymb}
\usepackage{graphicx}
\usepackage{cite}
\newcommand{\be}{\begin{equation}}
\newcommand{\ee}{\end{equation}}

\newcommand{\upd}{{\mathrm d}}

\newtheorem{lemma}{Lemma}

\newtheorem{remark}{Remark}

\begin{document}
\title{Asymptotic formula for quantum harmonic oscillator tunneling probabilities}
\author{Arkadiusz Jadczyk\smallskip \\ \emph{Laboratoire de Physique Th\'{e}orique, Universit\'{e} de Toulouse III \& CNRS}\\ \emph{118 route de Narbonne, 31062 Toulouse, France}\\
\emph{E-mail: ajadczyk@physics.org}}
\maketitle

\begin{abstract}Using simple methods of asymptotic analysis it is shown that for a quantum harmonic oscillator in $n$-th energy eigenstate the probability of tunneling into the classically forbidden region obeys an unexpected but  simple asymptotic formula: the leading term is inversely proportional to the cube root of $n$.\end{abstract}
\section{Introduction}
Consider the quantum harmonic oscillator in dimensionless variables, where the Hamiltonian, in terms of position and momentum operators $\hat{x},\hat{p},$ satisfying $[\hat{x},\hat{p}]=i,$ is given by the formula
\be \hat{H}=\frac{1}{2}(\hat{p}^2+\hat{x}^2).\notag\ee
Its normalized eigenstates, when represented as square integrable functions of the position variable $x$ are
\be \psi_n(x)=\pi^{-1/4}(2^n n!)^{-1/2} H_n(x)\,e^{-x^2/2},\notag\ee
$H_n$ are Hermite's polynomials: $H_n(x)=(-1)^ne^{x^2}\frac{d^n}{dx^n}\,e^{-x^2}.$ The corresponding probability densities $P_n(x)$ are then given by $|\psi_n(x)|^2,$ i.e.
\be P_n(x)=a_n \left(H_n(x)\,e^{-x^2/2}\right)^2,\notag\ee
where
\be a_n=\frac{1}{\pi^{1/2}2^n n!},\label{eq:an}\ee
so that the normalization condition holds:
\be \int_{-\infty}^{\infty}P_n(x)\upd x=1.\notag\ee
The functions $P_n(x)$ are symmetric with respect to the origin $x=0.$ The classical turning points are at $x=\pm\sqrt{2n+1},$ thus the probability of quantum tunneling into the classically forbidden region are given by the formula
\be P_{n,\mathrm{tun}}=2\int_{\sqrt{2n+1}}^{\infty}P_n(x)\,\upd x=2a_n Q_n,\label{eq:pn}\ee
where
\be Q_n=\int_{\sqrt{2n+1}}^{\infty}\left(H_n(x)\,e^{-x^2/2}\right)^2\,\upd x.\label{eq:qn}\ee
The probability distributions $P_n(x)$ are usually represented (see, for example, \cite[p. 168]{thaller}) as in Fig. \ref{fig:prn}
 \begin{figure}[!htb]
  \center
 \includegraphics[width=0.9\textwidth]{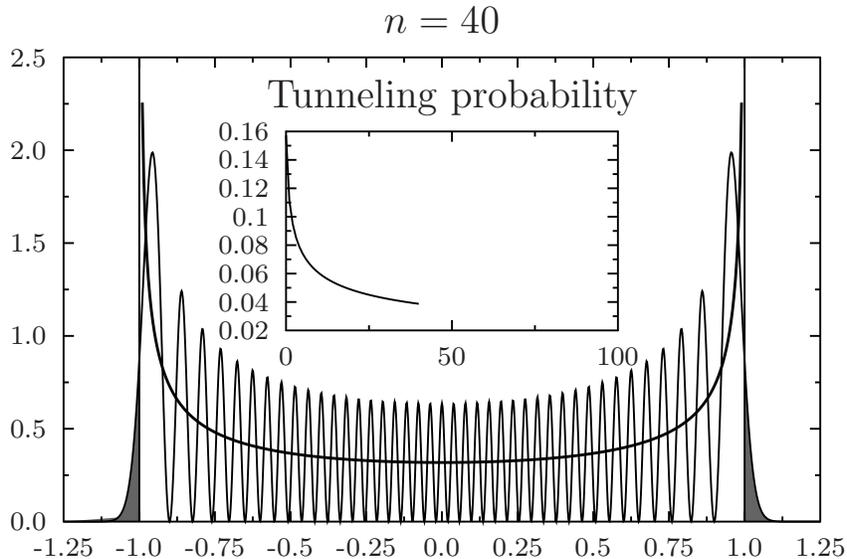}
 \caption{Probability distribution $P_{40}(x),$ with the classical turning points scaled to $x=\pm 1,$ and the curve $P_{n,\mathrm{tun}}$ for $n$ from $0$ to $40.$  The black $U$-shaped curve represents classical probability distribution $P_{\mathrm{class}}(x)=1/(\pi\sqrt{1-x^2}\,)$ for an oscillator with known energy and unknown phase. Shadowed tails beyond the classical turning points are responsible for non-zero quantum tunneling probabilities. }\label{fig:prn}
\end{figure}

The tunneling probabilities $P_{n,\mathrm{tun}}$ are rarely discussed in textbooks on quantum mechanics. Sometimes they are discussed in quantum chemistry textbooks. For instance in \cite[p. 92]{muller} the first three are calculated (though the second one with an error: it should be 0.1116 instead of 0.116). In \cite[p. 66-67; p. 98-99 in 2nd ed (2014)]{friedman} $P_{0,\mathrm{tun}}$ is calculated with a comment:  ``\textit{The probability of being found in classically forbidden regions decreases quickly with increasing  $n,$ and vanishes entirely as  $n$ approaches infinity}.'' Yet the question of ``how quickly" is left open there. To answer this question we will derive the asymptotic behavior, first by using a heuristic method based on a simple formula by Szeg\"{o}, then, in a rigorous way, using another formula due to Olver. Both methods lead to the same final result. The problem that we solve is closely related to a more general problem of large quantum numbers and the correspondence principle.  As can be seen from Refs. \cite{kiwi,gao,eltschka,boisseau},  even for such a simple quantum system as the harmonic oscillator, there are surprises (a more general discussion can be found in the monograph \cite{bolivar}).  Our result showing a very slow approach to zero of $P_{n,\mathrm{tun}}$ for large $n$ may be considered as another surprise.

In section \ref{sec:szego} we derive the leading term formula using heuristic reasoning. In section \ref{sec:olver} we use mathematically rigorous arguments based on dominated convergence theorem, while controlling the order of terms. In section \ref{sec:2order} we derive the
second order correction to the leading term. The proof of the main lemma used in this paper is given in the Appendix. Our theoretical results are also verified numerically.
\section{Derivation of the asymptotic tunneling probability formula\label{sec:szego}}
We are interested in the approximate behavior of $P_{n,\mathrm{tun}},$ as a function of $n$, for ``large'' $n.$ To this end we first use the asymptotic formula by Szeg\"{o} \cite[p. 201, (8.22.14)]{szego} ( valid in the transition regions )
\be e^{-\frac{x^2}{2}}H_n(x)=\frac{3^{\frac{1}{3}}2^{\frac{n}{2}+\frac{1}{4}}(n!)^{\frac{1}{2}}}{\pi^{\frac{3}{4}}n^{\frac{1}{12}}} \mathrm{A}(t)\left(1+O(n^{-\frac{2}{3}})\right),\label{eq:sza}\ee
\be x=(2n+1)^{1/2}-2^{-1/2}3^{-1/3}n^{-1/6}t.\notag\ee
The function $\mathrm{A}(x)$ used by Szeg\"{o} relates to the Airy functions $\mathrm{Ai}(x)$ by the following formula (cf. \cite[p.194]{olver}):
\be \mathrm{A}(t)=3^{-1/3}\pi\, \mathrm{Ai}(-3^{-1/3}t).\notag\ee
Thus Eq. \ref{eq:sza} can be written as\footnote{Another use of the Airy function in a discussion of the harmonic oscillator wave functions can be found, for instance, in \cite[Fig. 5-23]{powell}, in the section on WKB approximation.}
\be e^{-x^2/2}H_n(x)\approx \pi^{1/4}2^{n/2+1/4}(n!)^{1/2}n^{-1/12}\mathrm{Ai}(s),\label{eq:l}\ee
where \be s=2^{1/2}n^{1/6}(x-\sqrt{2n+1}).\notag\ee
The above asymptotic formula is proven to be valid for $t$ bounded.  In the calculation below I will use this result for all $s > 0,$ because the main contributions to integral come from a very small neighborhood of $s = 0,$ in particular when $n$ is large in (\ref{eq:l}).\\
Since $\frac{dx}{ds}=2^{-1/2}n^{-1/6},$ we can calculate now $Q_n$ defined in Eq. \ref{eq:qn} as
\be Q_n=\int_{\sqrt{2n+1}}^{\infty}\left(H_n(x)\,e^{-x^2/2}\right)^2\,\upd x.\label{eq:qn1}\ee
We get
\be Q_n\approx\pi^{1/2}2^n n! n^{-1/3}\int_0^{\infty}\mathrm{Ai}^2(s)\upd s.\notag\ee
The integral of the squared Airy function is known (cf. e.g. \cite[Eq. (3.75), p. 50]{airy}):
\be \int_0^{\infty}\mathrm{Ai}^2(s)\upd s=\frac{1}{3^{2/3}\,\Gamma^2[1/3]}.\notag\ee
Therefore (cf. Eqs. (\ref{eq:pn}),(\ref{eq:an}))
\begin{align} P_{n,\mathrm{tun}}&=2a_nQ_n&\\
&\approx 2\frac{1}{\pi^{1/2}2^n n!}\pi^{1/2}2^n n! n^{-1/3}\frac{1}{3^{2/3}\,\Gamma^2[1/3]}&\notag\\
&=\frac{2}{3^{2/3}\Gamma^2[1/3]}n^{-1/3}.\label{eq:pns}\end{align}
Thus, numerically,
\be P_{n,\mathrm{tun}}\approx \frac{0.133975}{n^{1/3}}.\label{eq:fita}\ee
 \begin{figure}[!htb]
  \center
 \includegraphics[width=0.9\textwidth]{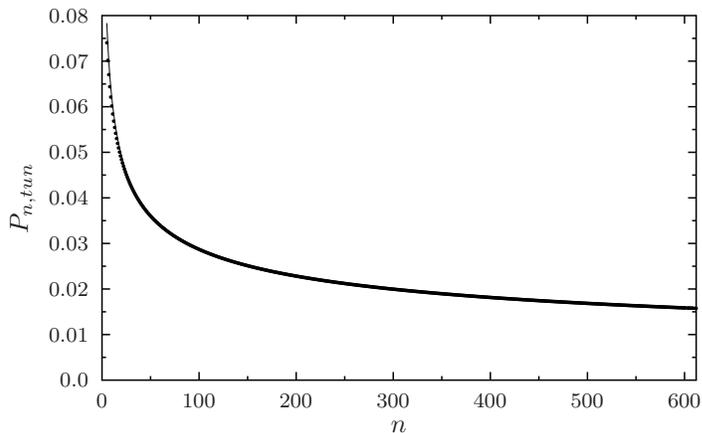}
 \caption{Comparison of the tunneling probabilities calculated from the exact formula with the curve given by Eq.
 (\ref{eq:fita}), starting from $n=5.$} \label{fig:fit1}
\end{figure}
As it can be seen from Fig. \ref{fig:fit1}, the asymptotic formula (\ref{eq:fita}) seems to agree with the numerical calculations for $n$ up to $612.$ Yet in its derivation we have assumed the validity of Szeg\"{o}'s asymptotic formula in the unbounded interval. Can such an extension be mathematically justified? We will address this question in the next section.
\section{Using Olver's formula\label{sec:olver}}
Olver \cite[p. 403]{olver97} gives the following formula valid for $x\ge 1$:
\be e^{-\frac{\nu^2x^2}{2}}H_n(\nu x)=c_n\left(\frac{\zeta}{x^2-1}\right)^{\frac{1}{4}}\left(\mathrm{Ai}(\nu^{\frac{4}{3}}\zeta)+\epsilon(x)\right),\label{eq:olver}\ee
where
\be\nu=\sqrt{2n+1},\quad c_n=(2\pi)^{\frac{1}{2}}e^{-\frac{\nu^2}{4}}\nu^{\frac{(3\nu^2-1)}{6}},\ee
\be \zeta=\left(\frac{3}{4}x(x^2-1)^{1/2}-\frac{3}{4}\cosh^{-1}(x)\right)^{2/3},\label{eq:zeta}\ee
and
\be |\epsilon(x)|\le 1.36\left(\exp(0.09\nu^{-2})-1\right)\mathrm{Ai}(\nu^{4/3}\zeta).\label{eq:eps}\ee
\begin{remark}: The function $\cosh^{-1}(x)$ in the above formula denotes the positive branch of the inverse function of $\cosh.$ For general aspects of asymptotic methods for integrals see \cite{temme}, where, in particular, Section 23.4 provides a complete description of the Airy-type expansion of Hermite polynomials. Notice also that $\nu^{-2}=1/(2n+1) = (1/(2n)) (1-1/(2n))+O(n^{-3}).$
\end{remark}
We will need Eq. (\ref{eq:olver}) squared.
Using Eqs. (\ref{eq:olver})-(\ref{eq:eps}) we write it in the form
\be \left(e^{-\frac{\nu^2x^2}{2}}H_n(\nu x)\right)^2= d_n \left(1+O(\frac{1}{n})\right) \left(\frac{\zeta}{x^2-1}\right)^{1/2}\mathrm{Ai}^2(\nu^{4/3}\zeta),\notag\ee
where
\be d_n=2\pi e^{-\nu^2/2}\nu^{(3\nu^2-1)/3}.\notag\ee
In order to use this formula for calculating $P_{n,\mathrm{tun}},$  we use  Eq. (\ref{eq:qn}) changing the integration variable to get
\begin{align} P_{n,\mathrm{tun}}&=2a_n\nu\int_1^{\infty}\left(H_n(\nu x)\,e^{-\nu^2 x^2/2}\right)^2 \upd x&\notag\\
&= 4 a_n\pi \nu e^{-\nu^2/2}\nu^{(3\nu^2-1)/3} (1+O(\frac{1}{n}))I_n,\notag\end{align}
where
\be I_n=\int_1^\infty \left(\frac{\zeta}{x^2-1}\right)^{1/2}\mathrm{Ai}^2(\nu^{4/3}\zeta) \upd x.\notag\ee
An easy power expansion leads to
\be 4 a_n\pi \nu e^{-\nu^2/2}\nu^{(3\nu^2-1)/3} =n^{1/3}\left(2^{7/3}+ O(\frac{1}{n})\right),\notag\ee
therefore
\be P_{n,\mathrm{tun}}= 2^{7/3}n^{1/3}\left(1+O(\frac{1}{n})\right)I_n.\notag\ee
\begin{figure}[!htb]
  \center
 \includegraphics[width=0.9\textwidth]{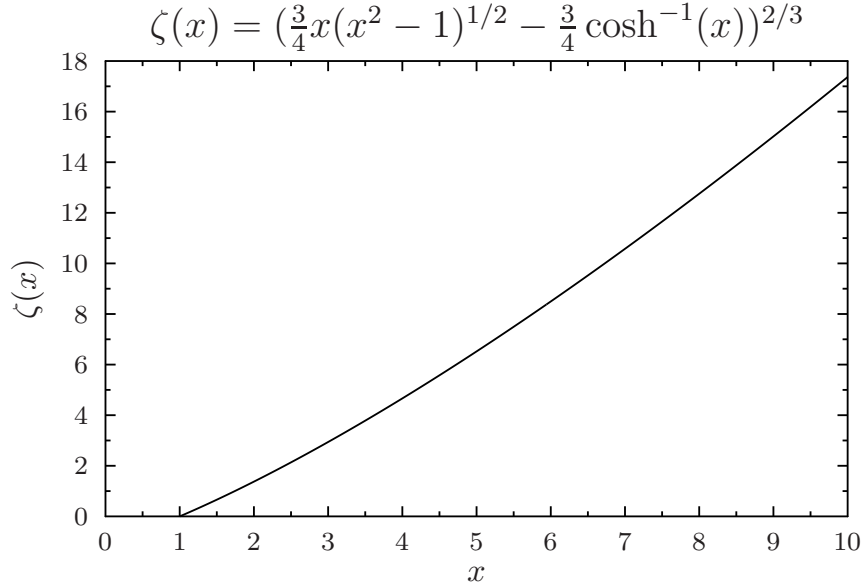}
 \caption{Plot of the function $\zeta(x)$ for $x>1.$}\label{fig:zet}
\end{figure}
It is convenient now to introduce a new integration variable \be t=\nu^{4/3}\zeta.\label{eq:t}\ee Fig.~\ref{fig:zet} shows the dependence of $\zeta$ on $x.$ We find
\be \frac{\upd\zeta}{\upd x}=\left(\frac{x^2-1}{\zeta}\right)^{1/2},\notag\ee
therefore
\be \upd x=\nu^{-4/3}\left(\frac{\zeta}{x^2-1}\right)^{1/2}\,\upd t, \notag\ee
and thus
\be I_n=\nu^{-4/3}\int_0^\infty f_n(t)\mathrm{Ai}^2(t)\upd t,\notag\ee
where
\be f_n(t)=\frac{\zeta}{x^2-1}\notag\ee is implicitly defined through $\zeta=\nu^{-\frac{4}{3}}t$ and Eq. (\ref{eq:zeta}) that allows us to find $x$ as a function of $\zeta.$
Since \be
\nu^{-4/3}=(2n+1)^{-2/3}= 2^{-2/3}n^{-2/3}\left(1+O(\frac{1}{n})\right),\notag\ee for $P_{n,\mathrm{tun}}$ we get
\be P_{n,\mathrm{tun}}= 2^{5/3}n^{-1/3}\left(1+O(\frac{1}{n})\right)\int_0^\infty f_n(t)\mathrm{Ai}^2(t)\upd t.\label{eq:pnf}\ee
Denote by $F_n$ the integral in Eq. (\ref{eq:pnf}):
\be F_n=\int_0^\infty f_n(t)\mathrm{Ai}^2(t)\upd t.\label{eq:Fn}\ee

We will show now that \be F_n\sim F_\infty=2^{-2/3}\int_0^\infty \mathrm{Ai}^2(t)\upd t=\frac{1}{2^{2/3}3^{2/3}\Gamma(\frac{1}{3})^2},\label{eq:fnlim}\ee
and therefore
\be P_{n,\mathrm{tun}}\sim\frac{2}{3^{2/3}\Gamma(\frac{1}{3})^2}\,n^{-1/3},\notag\ee
in agreement with Eq. (\ref{eq:pns}).
 \begin{figure}[!htb]
  \center
 \includegraphics[width=0.9\textwidth]{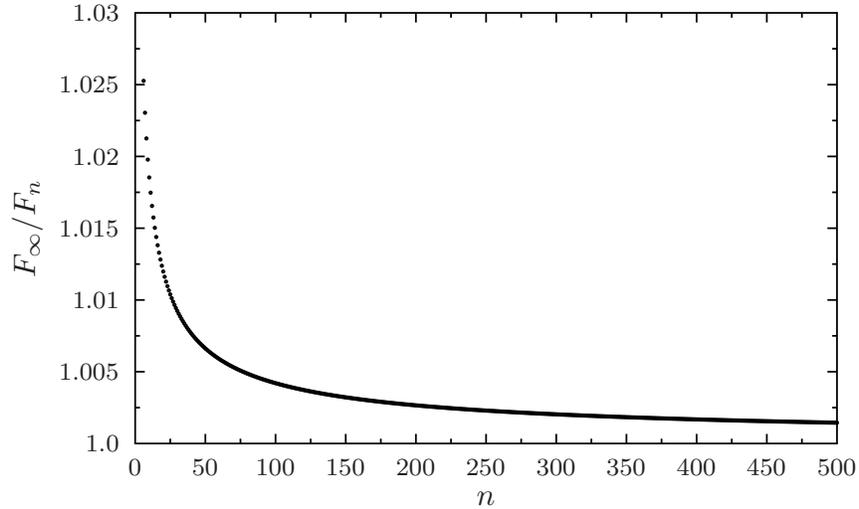}
 \caption{Numerically calculated ratios $F_\infty/F_n,$ $6\leq n\leq 500.$}\label{fig:r}
\end{figure}
Fig. \ref{fig:r} shows numerically calculated ratios $F_\infty/F_n$ for $6\leq n\leq 500.$
In order to verify Eq. (\ref{eq:fnlim}), we notice that (see the Appendix) the function \be f(x)=\frac{\left(\frac{3}{4}x(x^2-1)^{1/2}-\frac{3}{4}\cosh^{-1}(x)\right)^{2/3}}{x^2-1},\quad (x>1)\label{eq:s11}\ee
is monotonically decreasing, from $f(1)=\lim_{x\rightarrow 1+}\,f(x)=2^{-2/3}$ to $0$ at $x\rightarrow \infty.$
It follows that the sequence $f_n(t)$ is uniformly bounded. Moreover, with fixed $t>0,$ and since $\zeta=(2n+1)^{-1/3}t,$ when $n\rightarrow\infty,$ $\zeta\rightarrow 0,$ thus $x\rightarrow 1,$ therefore $\lim_{n\rightarrow\infty}f_n(t)= f(1)=2^{-2/3}.$ Thus \be \lim_{n\rightarrow\infty}\,F_n=\int_0^\infty \left(\lim_{n\rightarrow\infty}f_n(t)\right)\mathrm{Ai}^2(t)\upd t.\notag\ee
by the dominated convergence theorem.
\section{Next order approximation\label{sec:2order}}
Using the method described in the previous section one can easily get the next term in the asymptotic formula for $F_n.$ The result is
\be F_n=2^{-2/3}\int_0^\infty \mathrm{Ai}^2(t)\upd t-\frac{2^{-2/3}n^{-2/3}}{5}\int_0^\infty t\mathrm{Ai}^2(t)\upd t+O(\frac{1}{n}),\label{eq:fn2}\ee
which, after some simplifications leads to the formula
\be P_{n,\mathrm{tun}}= \frac{1}{n^{1/3}}\left(0.133975-\frac{0.0122518}{n^{2/3}}\right)+O(n^{-4/3}).\label{eq:2o}\ee
The second term in Eq. (\ref{eq:fn2}) can be derived from the second term in the Taylor series expansion $f_n(t)=f_n(0)+(df_n/dt)|_{t=0}\,t +O(t^2),$ if we use Eq. (\ref{eq:t}) and notice that
\be (df_n/dt)|_{t=0}=\nu^{-4/3}(df_n/d\zeta)|_{\zeta=0}\sim 2^{-2/3}n^{-2/3}(-1/5).\notag\ee

Fig. \ref{fig:fit2} shows the comparison of the exact data with the two curves. The black curve, obtained from Eq. (\ref{eq:2o}) is closer to the data than the grey one, obtained from the simple approximation given by Eq. (\ref{eq:pns}).
 \begin{figure}[!htb]
  \center
 \includegraphics[width=0.9\textwidth]{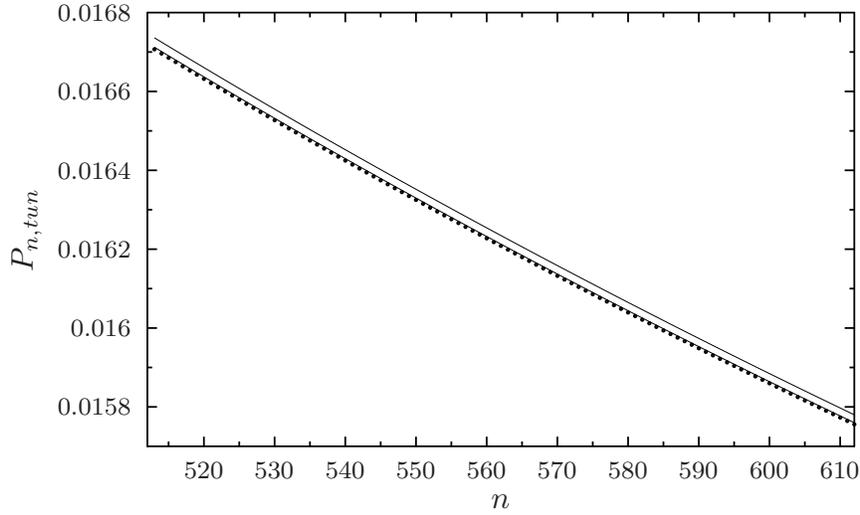}
 \caption{Details of the comparison for $n>512:$ exact data (dots), simple formula $P_{n,\mathrm{tun}}= 0.133975\, n^{-1/3}$ (upper curve), and the formula with the first correction term $P_{n,\mathrm{tun}}= 0.133975\,n^{-1/3}-0.0122518\, n^{-1}$ (lower curve).}  \label{fig:fit2}
\end{figure}
Additional correction terms have been obtained by Paris in \cite{rp} by different means.
\section{Appendix}
Fig. \ref{fig:zeta2} shows the computer generated plot of the function $\zeta(x)/(x^2-1)$ for $x>1.$ While the behavior of this function can be easily guessed from the plot, it is not that easy to prove rigorously that the function has the properties that are necessary for the application of the dominated convergence theorem used in this paper.
 \begin{figure}[!htb]
  \center
 \includegraphics[width=0.9\textwidth]{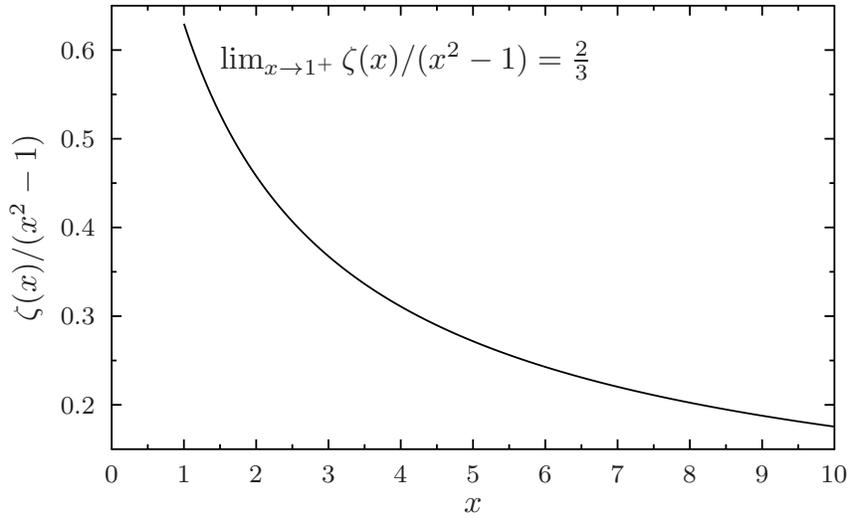}
 \caption{Plot of the function $\zeta(x)/(x^2-1).$}  \label{fig:zeta2}
\end{figure}
\begin{lemma}
Let $f(x)$ be as in Eq. (\ref{eq:s11}).
Then $f$ is monotonically decreasing, from $f(1)=\lim_{x\rightarrow 1+}\,f(x)=2^{-2/3}$ to $0$ at $x\rightarrow \infty.$
\label{lem:1}\end{lemma}
\begin{proof}
We have
\be f(x)=\left(\frac{3}{4}\tilde{f}(x)\right)^{\frac{2}{3}},\label{eq:s2}\ee
where
\be \tilde{f}(x)=\frac{x(x^2-1)^{1/2}-\cosh^{-1}(x)}{(x^2-1)^{3/2}},\notag\ee
and it is enough to show that the statement in the lemma holds for $\tilde{f}.$ Using the identity
\be \cosh^{-1}(x)=\log (x+\sqrt{x^2-1}),\notag\ee
and introducing $s,\, 0<s<\pi/2,$ through $x =\sec s,$ we can write $\tilde{f}(x)$ as a function $z(s)$ of $s$
\be z(s)= \frac{\sec s\, \tan\, s-\log(\sec s+\tan s)}{\tan^3 s}=\frac{a(s)}{b(s)},\label{eq:s3}\ee
and to prove the lemma we need to prove that $z$ is monotonically decreasing from $z(0)=\lim_{s\rightarrow 0+} z(s).$ Let \be h(u)=\sec u \tan^2 u.\notag\ee
By differentiation of the formulas below it is easy to verify that
\be
a(s)=  2 \int_0^s h(u)\upd u,\quad b(s)= 3 \int_0^s h(u)\,\sec u\,\upd u.
\notag\ee
We notice that $a(0)=b(0),$ and that, by de l'H\^{o}spital's rule $$ z(0)=\lim_{s\rightarrow 0+}\frac{a'(s)}{b'(s)}=\frac{2}{3\sec 0}=2/3.$$
Therefore, using Eqs. (\ref{eq:s11}),(\ref{eq:s2}), $f(1)=2^{-2/3}.$
Taking the derivative of Eq. (\ref{eq:s3}) we get
\be z'(s)=\frac{a'(s)b(s)-a(s)b'(s)}{b^2(s)}.\notag\ee The denominator is positive for $s>0.$ The numerator is negative owing to the fact that, for $s>0,$ \begin{align}
a'(s)b(s)-a(s)b'(s)&=6\left(h(s)\int_0^s h(u)\sec u\,\upd u-h(s)\sec s\int_0^s h(u)\upd u\right)&\notag\\
&=6h(s)\int_0^s h(u)(\sec u - \sec s)\upd u,
\notag\end{align}
with $h(s)>0$, and $\sec(u)<\sec s$  for $u<s.$ Therefore $z'(s)<0.$
\end{proof}
\section*{Acknowledgments}
The author would like to thank  Ilia Krasikov, Richard Paris, Philippe Sosoe, Jerzy Szulga, and Nico Temme for their encouragement, helpful insights, comments and criticism. In particular Nico Temme has suggested using Olver's formula
and Richard Paris assisted me throughout the evolution of this paper.  Thanks are also due to Ronald Friedman for his interest in the problem, as well as to the anonymous referee for helpful comments.


\begin{thebibliography}{0}
\bibitem{thaller}B. Thaller: {\em Visual Quantum Mechanics\/}, Springer-Verlag, New York 2000.
\bibitem{kiwi}G. G. Cabrera and M. Kiwi M: Large quantum-number states and the correspondence principle, {\em Phys. Rev.\/} A {\bf 36},2995(R) (1987). 
\bibitem{gao}B. Gao: Breakdown of Bohr's correspondence principle, {\em Phys. Rev. Let.\/} Phys. Rev. Lett.
{\bf 83}, 4225-4228 (1999). 
\bibitem{eltschka}C. Eltschka, H. Friedrich and M. J. Moritz: Comment on ``Breakdown of
Bohr's correspondence principle" {\em Phys. Rev. Lett.\/} {\bf 86}, 2693 (2001). 
\bibitem{boisseau}C. Boisseau, E. Audouard and J. Vigue: Comment on ``Breakdown of Bohr's
correspondence principle," {\em Phys. Rev. Lett.\/} {\bf 86}, 2694 (2001).
\bibitem{bolivar} A. O. Bolivar: {\em Quantum-classical correspondence: dynamical quantization and the classical limit\/}, Springer-Verlag, Berlin-Heidelberg 2004.
\bibitem{muller} M. Mueller: {\it Fundamentals of Quantum Chemistry\/}, Kluwer Academic Publishers, New York 2002.
\bibitem{friedman} P. Atkins, J. de Paula, R. Friedman: {\em Quanta, Matter and Change. A molecular approach to physical chemistry\/}, W. H. Freeman and Company, New York 2009.
\bibitem{szego} Gabor Szeg\"{o}, {\em Orthogonal Polynomials\/}, American Mathematical Society, Providence, Rhode Island  1939.
\bibitem{olver} F. W. J. Olver, D. W. Lozier, R. F. Boisvert, C. W. Clark: {\em NIST Handbook of Mathematical Functions\/}, NIST and Cambridge University Press, New York 2010.
\bibitem{powell} J. L. Powell, B. Crasemann: {\em Quantum Mechanics\/}, Addison-Wesley, London 1961.
 \bibitem{airy} O. Vall\'{e}e, M. Soares: {\em Airy Functions and Applications to Physics}, World Scientific, New Jersey 2004.
\bibitem{olver97}F. W. J. Olver: {\em Asymptotics and Special Functions\/}, A K  Peters, Wellesley, Massachusetts 1997.
\bibitem{temme} N. Temme: {\it Asymptotic Methods for Integrals}, World Scientific, Singapore 2014.
\bibitem{rp}R. Paris: {\it  Asymptotic evaluation of an integral arising in quantum harmonic oscillator tunnelling probabilities\/}, 	arXiv:1502.03382 [math.CA] (2015).
 \end{thebibliography}
\end{document}